\documentclass[conference,letterpaper]{IEEEtran}

\usepackage{amsmath}
\usepackage{amsfonts}
\usepackage{amssymb}
\usepackage{graphicx}
\usepackage{tikz}
\usepackage{ifthen}
\usepackage{xstring}
\usepackage{mathtools}
\usepackage[ruled,vlined,linesnumbered]{algorithm2e}
\usepackage{xcolor}
\usepackage{stmaryrd}
\usepackage{cite}
\usepackage{microtype}
\usepackage{xspace}
\usepackage{balance}
\usepackage{amsthm}

\allowdisplaybreaks

\newif\ifproof
\prooftrue

\newcommand{\calX}{\mathcal{X}}
\newcommand{\calY}{\mathcal{Y}}
\newcommand{\mimicGenie}{\mathrm{mimicGenie}}

\newcommand{\vecd}{\mathbf{d}}
\newcommand{\vecg}{\mathbf{g}}

\newcommand{\vecx}{\mathbf{x}}
\newcommand{\vecX}{\mathbf{X}}
\newcommand{\vecy}{\mathbf{y}}
\newcommand{\vecY}{\mathbf{Y}}
\newcommand{\vecz}{\mathbf{z}}

\newcommand{\vecZ}{\mathbf{Z}}

\newcommand{\xI}{\vecx_{\mathrm{I}}}
\newcommand{\xII}{\vecx_{\mathrm{II}}}
\newcommand{\zI}{\vecz_{\mathrm{I}}}
\newcommand{\zII}{\vecz_{\mathrm{II}}}

\newcommand{\concat}{\odot}

\newtheorem{theo}{Theorem}
\newtheorem{lemm}[theo]{Lemma}

\SetKwFunction{recursivelyCalcTrellis}{recursivelyCalcTrellis}
\SetKwFunction{makeDecision}{makeDecision}

\newcommand{\floor}[1]{\lfloor #1 \rfloor}

\newcommand{\eqann}[2][=]{\overset{\mathclap{(\text{#2})}}{#1}} 
\newcommand{\eqannref}[1]{$(\text{#1})$}

\newcommand{\errh}[1][h]{\mathrm{err}_{#1}}

\SetAlFnt{\small}
\SetKwInOut{Require}{Require} 
\SetKwFunction{getArrayPointerP}{getArrayPointer\_P}
\SetKwFunction{getArrayPointerC}{getArrayPointer\_C}
\SetKwFunction{assignInitialPath}{assignInitialPath}
\SetKwFunction{pathIndexInactive}{pathIndexInactive}
\SetKwFunction{continuePathsFrozenBit}{continuePaths\_FrozenBit}
\SetKwFunction{continuePathsUnfrozenBit}{continuePaths\_UnfrozenBit}
\SetKwFunction{recursivelyCalcP}{recursivelyCalcP}
\SetKwFunction{recursivelyUpdateB}{recursivelyUpdateB}
\SetKwFunction{recursivelyUpdateC}{recursivelyUpdateC}
\SetKwFunction{killPath}{killPath}
\SetKwFunction{clonePath}{clonePath}
\SetKwFunction{findMostProbablePath}{findMostProbablePath}
\SetKwFunction{initializeDataStructures}{initializeDataStructures}
\SetKwData{forksArray}{forksArray}
\SetKwData{contForks}{contForks}
\SetKwData{probForks}{probForks}
\SetKwData{inactivePathIndices}{inactivePathIndices}
\SetKwData{inactiveArrayIndices}{inactiveArrayIndices}
\SetKwData{arrayReferenceCount}{arrayReferenceCount}
\SetKwData{pathIndexToArrayIndex}{pathIndexToArrayIndex}
\SetKwData{arrayPointerP}{arrayPointer\_P}
\SetKwData{arrayPointerC}{arrayPointer\_C}
\SetKwData{activePath}{activePath}
\SetKw{Continue}{continue}
\SetKw{New}{new}
\SetKw{Push}{push}
\SetKw{Pop}{pop}
\SetKw{False}{false}
\SetKw{True}{true}
\SetKw{myand}{and}
\SetKw{myor}{or}
\DontPrintSemicolon
\SetAlgoSkip{-5pt}

\IEEEoverridecommandlockouts

\newcommand{\emptystring}{\epsilon}

\newcommand{\Ppadleft}{P_{\mathrm{padLeft}}}
\newcommand{\Ppadright}{P_{\mathrm{padRight}}}
\newcommand{\vecymiddle}{\vecy_{\mathrm{middle}}}
\newcommand{\vecyleft}{\vecy_{\mathrm{left}}}
\newcommand{\vecyright}{\vecy_{\mathrm{right}}}
\newcommand{\vecYleft}{\vecY_{\mathrm{left}}}
\newcommand{\vecYright}{\vecY_{\mathrm{right}}}
\newcommand{\vecYmiddle}{\vecY_{\mathrm{middle}}}
\newcommand{\veczero}{\mathbf{0}}
\newcommand{\vecone}{\mathbf{1}}

\newcommand{\pdel}{{p_\text{d}}}
\newcommand{\pins}{{p_\text{i}}}
\newcommand{\psub}{{p_\text{s}}}
 
\newcommand{\twobibs}[2]{#2} 

\begin{document}
\title{Polar Codes for Channels with \\ Insertions, Deletions, and Substitutions}
\date{}
\author{%
    \IEEEauthorblockN{Henry D. Pfister} \IEEEauthorblockA{Duke University} \and
    \IEEEauthorblockN{Ido Tal} \IEEEauthorblockA{Technion} 

}

\maketitle
\begin{abstract}
    This paper presents a coding scheme for an insertion deletion substitution
     channel. We extend a previous scheme for the deletion channel where polar codes are modified by adding ``guard bands'' between segments. In the new scheme, each guard band is comprised of a middle segment of `1' symbols, and left and right segments of `0' symbols. Our coding scheme allows for a regular hidden-Markov input distribution, and achieves the information rate between the input and corresponding output of such a distribution. Thus, we prove that our scheme can be used to efficiently achieve the capacity of the channel. The probability of error of our scheme decays exponentially in the cube-root of the block length.

\end{abstract}

\section{Introduction}
In many communications systems, symbol-timing errors can result in insertion and deletion errors.
For example, the insertion deletion substitution (IDS) channel maps a length-$N$ input string to a finite output string by sampling an i.i.d.\ output process for each input that selects between insertion, deletion, and substitution.
These types of channels were first studied in the 1960s~\cite{Gallager_1961,Dobrushin_1967} and modern coding techniques were first applied to them in~\cite{Davey_2001}.
Over the past 15 years, bounds on the capacity of synchronization error channels have been significantly improved~\cite{Mitzenmacher_2009,Kanoria_2013,Rahmati_2015,Castiglione_2015,Cheraghchi_2019}.

In~\cite{Tal-isit19,TPFV:20a}, a capacity-achieving coding scheme is introduced for the deletion channel based on polar codes.
The construction and proof builds upon many earlier results (e.g., \cite{Wang_2014,Wang_2015,Thomas_2017,Tian_2017,Tian_2018,Tian_it2018,Li_2019,SasogluTal:18a,ShuvalTal:18a}); see~\cite{TPFV:20a} for a detailed description of these connections.

The construction in~\cite{Tal-isit19,TPFV:20a} is based on generating codewords consisting of smaller blocks separated by guard bands.
After reception, the overall output sequence is separated into blocks associated with the smaller input blocks.
However, the separation process changes the effective channel experienced by the small blocks.
In particular, the guard bands are long blocks of zeros and the separation process removes all zeros on either side of the small block.
The analysis in~\cite{Tal-isit19,TPFV:20a} shows that the resulting channel, dubbed the trimmed deletion channel (TDC), polarizes weakly and has the same mutual information rate as the original deletion channel.
Due to the possibly unbounded memory in the deletion channel, the standard extension~\cite{Shuval_Tal_Memory_2017} to strong polarization does not work.
Instead, strong polarization can be shown for the polar combining of these small blocks due to the independence provided by the guard bands.
These elements complete the achievability proof for the deletion channel.

In this paper, we apply roughly the same coding scheme to the IDS channel.
The main difference is that separating the overall output sequence into smaller blocks is more challenging.
For the deletion channel, an input that only contains zeros always gives an output that only contains zeros.
Thus, the separation process consists of parsing into small blocks and removing zeros from the edges.
For the IDS channel, an input only containing zeros typically gives an output containing both zeros and ones.
Fortunately, the expected fraction of zeros will be noticeably larger than the fraction of ones.
This observation along with a more complicated parsing process can be used to separate the overall output sequence into blocks associated with the smaller input blocks.

The key challenge is designing the parsing process so that the effective channel experienced by the small block can be analyzed.
In particular, our parsing process produces segments that can be seen roughly as the IDS output of an input consisting of a prefix of zeros, the original input data, and a suffix of zeros.
For all the small output blocks, the prefix and suffix lengths are i.i.d.\ random variables with a known distribution.
We refer to the resulting channel as the dirty zero-padded (DZP) IDS channel.
To establish the coding theorem for the IDS channel, we must show three things.
First, that our parsing of the IDS channel output gives small blocks whose joint input-output distribution matches that of the DZP channel.
Second, that the DZP channel polarizes weakly and has the same mutual information rate as the original IDS channel.
Third, that the trellis representation of the joint input-output distribution of the IDS channel~\cite{Davey_2001} can be modified to give the joint input-output distribution of the DZP channel.
In this work, we establish these three elements and describe the first two elements herein.

By combining the parsing process described in this paper with the results of~\cite{Tal-isit19,TPFV:20a}, one gets the following theorem. Due to space limitations, many details are deferred to the extended version of this paper.
\begin{theo}\label{theo:main}
    Fix a regular hidden-Markov input process and a parameter $\nu \in (0,1/3]$. The rate of our coding scheme approaches the mutual information rate between the input process and the binary IDS channel output. The encoding and decoding complexities are $O(\Lambda \log \Lambda)$ and $O(\Lambda^{1+3\nu})$, respectively, where $\Lambda$ is the blocklength. For any $0 < \nu' < \nu$ and sufficiently large blocklength $\Lambda$, the probability of decoding error is at most $2^{-\Lambda^{\nu'}}$.
\end{theo}

The structure of this paper is as follows. In Section~\ref{sec:channels} we define the IDS channel, and also a close variant which we term the ``dirty zero padding IDS channel'' (DZP). Section~\ref{sec:encoding} details how encoding is done. In Section~\ref{sec:decoding}, we define two decoding methods. Namely, we first define a decoding method executed by a genie, which is in possession of some extra information (it knows where the ``commas'' which separate the outputs corresponding to certain input blocks are). The utility of the genie's decoding method is that it is easy to analyze (the DZP channel is used in the analysis). We then a define second decoding method: Aladdin's decoding method. Since Aladdin is a mere mortal, he does not have knowledge of where the above commas lie. That is, Aladdin's method is the one we can actually implement. The main trick is to show that with very high probability, the genie's decoder and Aladdin's decoder produce the exact same result.

\section{Channel models}
\label{sec:channels}

In this section, we define the IDS and DZP channels.

\subsection{Dobrushin's Channel and the IDS Channel}
\label{subsec:idsChannel}

In 1967, Dobrushin introduced a general class of channels with synchronization errors and proved a random coding theorem for that class~\cite{Dobrushin_1967}.
The model consists of a finite input alphabet $\mathcal{X}$ and a conditional distribution $p_{Y|X} (\cdot|x)$ over finite output strings $\mathcal{Y} \subseteq  \calX^* = \cup_{n=0}^\infty \mathcal{X}^n$ given $x\in \mathcal{X}$, where $\emptystring$ denotes the empty string of length $0$.
For the input $\vecx = (x_1,x_2,\ldots,x_N) \in \mathcal{X}^N$, the channel output is generated by drawing $Y_n \sim p_{Y|X}(\cdot|x_n)$ i.i.d.\ and concatenating to get
$$ \vecY = Y_1 \concat Y_2 \concat \cdots \concat Y_N .$$

For example, the binary deletion channel with deletion probability $\pdel$ has $\mathcal{X} = \{0,1\}$ and $\mathcal{Y} = \{\emptystring,0,1\}$ with non-zero probabilities $P_{Y|X}(\emptystring|x)=\pdel$ and $P_{Y|X}(x|x)=1-\pdel$ for all $x\in \mathcal{X}$.
Similarly, the binary IDS channel we consider has IDS probabilities $(\pins, \pdel, \psub)$, $\calX = \{0,1\}$,
and $\calY = \left\{ \emptystring, 0,1, 00, 01, 10, 11 \right\}$ with non-zero probabilities
$P_{Y|X}(\emptystring|x)=\pdel$,
$P_{Y|X}(x|x)=1-\pins-\pdel-\psub$,
$P_{Y|X}(\overline{x}|x)=\psub$,
$P_{Y|X}(0x|x) = P_{Y|X}(1x|x)=\pins/2$ for all $x\in \mathcal{X}$.



While we focus on this binary IDS channel for concreteness, the approach described here should generalize to any well-behaved binary-input Dobrushin channel for which the output distribution associated with the all-zero input is distinguishable from finite shifts of the output distribution associated with the all-one input.
For simplicity, we focus on the case where they are distinguishable simply by counting ones and zeros.

Define $\alpha_{0|x}$ ($\alpha_{1|x}$) as the expected number of $0$ ($1$) symbols at the output of the channel, given that the input was $x \in \calX$. Note that the expected length of an output, given that the input was $x \in \calX$ is $\alpha_{0|x} + \alpha_{1|x}$. We require that this sum is independent of $x$, and denote it as
\begin{equation}
    \label{eq:beta}
    \beta = \alpha_{0|0} + \alpha_{1|0} = \alpha_{0|1} + \alpha_{1|1} \; .
\end{equation}

We also require an ``advantage'' to $x$ at the output, if the input was $x$. That is, we require that
\begin{equation}
    \label{eq:advantageToInput}
    \alpha_{0|0} > \alpha_{1|0} \qquad \mbox{and} \qquad \alpha_{1|1} > \alpha_{0|1}  \; ,
\end{equation}
and denote
\begin{equation}
    \label{eq:gamma}
    \gamma \triangleq \frac{\min\{ \alpha_{0|0} - \alpha_{1|0}, \alpha_{1|1} - \alpha_{0|1} \}}{2} > 0 \; ,
\end{equation}
where the inequality follows by (\ref{eq:advantageToInput}).

Informally, the above ``advantage to the input at the output'' will allow us to differentiate between a long input of $0$ symbols and a long input of $1$ symbols. Specifically, fix a window length $h>0$ and an $x\in \mathcal{X}$.
Then, generate an output sequence of length at least $h+1$ and optionally remove the first output bit.
Then, we count the number of $0$ symbols contained in the first $h$  positions of the string. If it is at least $h/2$, then we declare that $x=0$; otherwise, we declare that $x=1$. The following lemma states that we have a very high chance of guessing correctly, for large enough window length $h$.


\begin{lemm}
    \label{lemm:windowAdversary}
    Let $x \in \calX$ be fixed, and let a window length $h \geq h_0$ be given, where $h_0$ is a constant dependent on the channel. Let $\vecY$ be a string of length $h$ generated by truncating the output associated with the all-$x$ input where the first output bit is optionally removed. Then, the probability that $\vecY$ contains fewer than $h/2$ bits equal to $x$ is less than
    \[
        \errh \triangleq e^{-h \cdot c_0} \; ,
    \]
    where $c_0$ is a positive constant dependent on the channel.
\end{lemm}
\begin{IEEEproof}
See Appendix.
\end{IEEEproof}

\subsection{Dirty-Zero-Padding IDS channel}
The DZP channel $W^\star$ is defined by the IDS channel $W$, the input blocklength $N_0$, and two probability distributions over $\mathcal{X}^*$, $\Ppadleft$ and $\Ppadright$. Given the length-$N_0$ input string $\vecx$, we first pass $\vecx$ through the IDS channel $W$ and let $\vecymiddle$ denote the output. Next, we draw two independent vectors $\vecyleft$ and $\vecyright$ according to the probability distributions $\Ppadleft$ and $\Ppadright$, respectively. The output of the DZP channel is then given by
\begin{equation}
    \label{eq:ystar_yleftMiddleRight}
    \vecy^\star = \vecyleft \concat \vecymiddle \concat \vecyright \; .
\end{equation}

We will specify $\Ppadleft$ and $\Ppadright$ later. For now, let us say informally that $\vecyleft$ and $\vecyright$ are the result of passing strings of `0' symbols through the channel $W$. Hence the name: we pad $\vecy$ from the left and right by vectors corresponding to zeros ``dirtied'' by passing through the channel $W$.

Informally, the following lemma states that, in the limit  as $N_0 \to \infty$, the mutual information rates of $W$ and $W^\star$ are equal. As will become apparent later, the maximum possible length of $\vecyleft$ and the maximum possible length of $\vecyright$ both grow sub-linearly in $N_0$. Hence, the condition of the lemma is not vacant.

\begin{lemm}
\label{eq:IDS_DZP_mutualInformation}
    Let $\vecX \in \calX^{N_0}$ be a random vector of length $N_0$. Let $\vecY$ and $\vecY^\star$ be the outputs gotten when $\vecX$ is input to the IDS channel $W$ and the DZP channel $W^\star$, respectively. Let $N_0$ be large enough so that the maximum length that $\vecyleft$ can take and the maximum length that $\vecyright$ can take are both at most $N_0-1$. Then,
    \[
        \frac{I(\vecX;\vecY)- 2 \log_2 N_0}{N_0}  \leq \frac{I(\vecX;\vecY^\star)}{N_0} \leq \frac{I(\vecX;\vecY)}{N_0}\; .
    \]
\end{lemm}
\begin{proof}
    It suffices to prove the inequalities for the numerators since the denominators all equal $N_0$.
          The inequality $I(\vecX;\vecY^\star) \leq I(\vecX;\vecY)$ follows by the data-processing inequality, since $\vecX$, $\vecY$, and $\vecY^\star$ form a Markov chain, in that order. We will show that
          \begin{equation} \label{eq:dzp_info_bound}
               I(\vecX;\vecY)- 2 \log_2 N_0 \leq I(\vecX;\vecY^\star) \; .
          \end{equation}
          Let us first denote $\vecY^\star = \vecYleft \concat \vecYmiddle \concat \vecYright$, as described above. Next, note that we can assume w.l.o.g.\ that $\vecY = \vecYmiddle$. Finally, note that $\vecX$, $(\vecY^\star, |\vecYleft|, |\vecYright|)$, and $\vecY$ form a Markov chain, in the order, where $|\cdot|$ denotes the length of a string. Thus,
          \begin{IEEEeqnarray*}{rCl}
              I(\vecX;\vecY)  &\leq& I(\vecX;\vecY^\star,|\vecYleft|, |\vecYright|) \\
                              & \leq & I(\vecX;\vecY^\star) + H(|\vecYleft|) +  H(|\vecYright|) \\ 
                              & \leq & I(\vecX;\vecY^\star) + 2 \log_2 (N_0) \; ,
          \end{IEEEeqnarray*}
          because both $|\vecYleft|$ and $|\vecYright|$ can take at most $N_0$ different values. Thus, \eqref{eq:dzp_info_bound} holds and the proof is complete.
\end{proof}

\section{Encoding}
\label{sec:encoding}
Suppose for a moment that we were coding not for the IDS channel $W$, but for the DZP channel $W^\star$.
First of all, recall that the channel $W^\star$ accepts a block of length $N_0$ bits. We choose a typically ``large'' $N_0$. However, instead of only sending a single block of length $N_0$, we send $\Phi$ such blocks, denoted $\vecx(1),\vecx(2),\ldots,\vecx(\Phi)$. The important point to note is the output: denote the output of $W^\star$ corresponding to $\vecx(i)$ as $\vecy^\star(i)$. We assume that the output corresponding to the above input is $(\vecy^\star(1),\vecy^\star(2),\ldots,\vecy^\star(\Phi))$, as opposed to $\vecy^\star(1) \concat \vecy^\star(2) \concat \cdots \concat \vecy^\star(\Phi)$. That is, we assume that the output blocks corresponding to the input blocks are \emph{punctuated}. Namely, given the output corresponding to $\Phi$ blocks, we can distinguish the output corresponding to input block $\vecx(i)$. This is in stark contrast to $W$, in which no such punctuation is given.

For this setting, one can both encode and decode using polar codes; this is very similar to what was done in~\cite{TPFV:20a} with the DZP channel playing the role of the block-TDC channel. Given the information symbols and frozen indices, the information symbols are mapped to a polar codeword of length $\Phi N_0$ using polar encoding. Also, extending the ideas in \cite{Davey_2001,TPFV:20a}, we can build a trellis for calculating the joint probability of $\vecx$ being the input to $W^\star$ and $\vecy^\star$ being the output (building such a trellis involves the use of $\Ppadleft$ and $\Ppadright$). Finally, using Lemma~\ref{eq:IDS_DZP_mutualInformation} and essentially the same proof as \cite{TPFV:20a}, this coding scheme can approach the capacity of the IDS channel $W$. Due to lack of space, we do not go into further details.

Our coding scheme for the IDS channel $W$ consists of two phases. In the first phase, we produce the blocks $\vecx(1),\vecx(2),\ldots,\vecx(\Phi)$ by taking the whole polar codeword and adding commas to separate into blocks of length $N_0$. Then, we imagine these blocks being transmitted over $W^\star$. In the second phase, we add guard bands (defined shortly) between the above blocks. The result is a long codeword that is transmitted over the channel $W$. Loosely speaking, the purpose of the guard bands is to allow the decoder to simulate the operation of $W^\star$ on the blocks $\vecx(1),\vecx(2),\ldots,\vecx(\Phi)$, even though we are in fact transmitting over the channel $W$.

Denote $N = 2^n$, $N_0 = 2^{n_0}$ and $\Phi=2^{n_1} = 2^{n-n_0}$, where $n_0 = \lceil n \nu \rceil$ and $\nu$ was fixed in Theorem~\ref{theo:main}. Let
\begin{equation}
    \label{eq:vecxAsConcatenationOfBlocks}
    \vecx = \vecx(1) \concat \vecx(2) \concat \cdots \concat \vecx(\Phi)
\end{equation}
be a vector of length $N$, consisting of blocks $\vecx(i)$, each of length $N_0$. We denote by $g(\vecx)$ the result of adding guard bands to $\vecx$. For this, let us denote $\vecx = \xI \concat \xII$, where $\xI$ and $\xII$ are the left and right halves of $\vecx$, each of length $N/2$.

\begin{align}
    \label{eq:guardBand}
    g(\vecx) &\triangleq
    \begin{cases}
    \vecx & \text{if }n \leq n_0
    \\
        g(\xI) \concat \vecg_n \concat g(\xII) & \text{if }n> n_0, 
    \end{cases} 
\end{align}
where $\vecg_n$ is termed the guard band and defined as follows. Denote by $\veczero(\ell)$ and $\vecone(\ell)$ a string of $\ell$ consecutive `$0$' symbols and a string of $\ell$ consecutive `$1$' symbols, respectively. Let
\[
     \ell_n \triangleq 2^{\floor{(1-\xi) (n-1)}} \label{eq:ln} \; ,
\]
\noindent where $\xi \in (0,1/2)$ is a `small' constant determined by the difference between $\nu$ and $\nu'$ in Theorem~\ref{theo:main}. Then,
\begin{equation}
    \label{eq:vecgnBlocks}
    \vecg_n \triangleq \underbrace{\veczero(\ell_{n_0})}_{\vecg_n^{\mathrm{left}}} \concat \overbrace{\underbrace{\vecone(\ell_n)}_{\vecg_n^{\mathrm{midleft}}} \concat \underbrace{\vecone(\ell_n)}_{\vecg_n^{\mathrm{midright}}}}^{\vecg_n^{\mathrm{mid}}} \concat \underbrace{\veczero(\ell_{n_0})}_{\vecg_n^{\mathrm{right}}} \; .
\end{equation}
We note that $\vecg_n^{\mathrm{left}}$ and $\vecg_n^{\mathrm{right}}$ are not, in fact, functions of $n$.

\section{Decoding}
\label{sec:decoding}
We now consider two settings for decoding. In both settings, a vector $g(\vecx)$ is transmitted over the IDS channel $W$, and the corresponding output is $\vecy$. Both settings differ only in their preliminary step, which parses the received vector $\vecy$ into $\Phi$ sub-vectors. In the first setting,  which we call ``genie parsing'', an all-knowing genie receives the output $\vecy$ and adds commas in certain appropriate places. Recall from (\ref{eq:vecxAsConcatenationOfBlocks}) that $\vecx$ is comprised of $\Phi$ blocks. After adding commas to the output, the genie produces for each block $\vecx(i)$ a corresponding output $\vecy^\star(i)$. The result is a series of outputs $\vecy^\star=(\vecy^\star(1), \vecy^\star(2),\ldots,\vecy^{\star}(\Phi))$, where for each $i$, the probability law of $\vecy^{\star}(i)$ given $\vecx^{\star}(i)$ is the DZP channel $W^\star(\vecy^\star(i)|\vecx(i))$. We then use the methods described in \cite{TPFV:20a} to decode $\vecx$ from $\vecy^\star$.

The second setting is called ``Aladdin parsing''. As before, $g(\vecx)$ is transmitted and $\vecy$ is received. The goal of Aladdin is to produce the same sequence $\vecy^\star$ as the genie. Since Aladdin is a mere mortal, he does not have the knowledge required to guarantee that he will add commas in the appropriate places.

This raises the question, ``Why does the genie output have dirty zero-padding?''.
An all-knowing genie could produce the IDS output sequences $\vecy(1),\ldots,\vecy(\Phi)$.
But, our genie chooses a weaker strategy (based on an i.i.d.\ dither sequence) so that Aladdin can hope to match the genie's parsing by making use of the guard bands.
Thus, we will show that Aladdin can succeed in producing $\vecy^\star$ with very high probability.


\subsection{Genie parsing}
Recall from (\ref{eq:guardBand}) that the codeword we transmit is comprised of blocks $\vecx(i)$, $1 \leq i \leq \Phi$, separated by guard bands. Let $\vecg(i)$ denote the guard band between $\vecx(i)$ and $\vecx(i+1)$, where $\vecg(i)$ equals $\vecg_{n'}$ for some $n_0 < n' \leq n$ which is a function of $i$. Now we recall from (\ref{eq:vecgnBlocks}) that each guard band $\vecg(i)$ is comprised of four blocks, which we denote
$\vecg^{\mathrm{left}}(i)$, $\vecg^{\mathrm{midleft}}(i)$, $\vecg^{\mathrm{midright}}(i)$, and $\vecg^{\mathrm{right}}(i)$, for $1 \leq i \leq \Phi - 1$. The genie receives the output $\vecy$, and adds commas between all the blocks because the genie can distinguish which substring of $\vecy$ equals $\vecy(i)$, the output corresponding to $\vecx(i)$. It can also distinguish which part of $\vecy$ corresponds to $\vecg^{\square}(i)$, where $\square \in \{ \mathrm{left}, \mathrm{midleft}, \mathrm{midright}, \mathrm{right} \}$. We denote the relevant part of $\vecy$ as $\vecd^{\square}(i)$, where ``d'' stands for ``dirty''.

Recall from (\ref{eq:ystar_yleftMiddleRight}) that, in order to return
\[
    \vecy^\star(i) = \vecyleft(i) \concat \vecy(i) \concat \vecyright(i) \; ,
\]
the DZP channel $W^\star$ must pad $\vecy(i)$ from the left and right. This padding is according to the probability distributions $\Ppadleft$ and $\Ppadright$, which have yet to be specified. Now, we define how the genie produces $\vecyleft(i)$ and $\vecyright(i)$ from the following punctuated segment of $\vecy^\star$,
\[
    \vecd^{\mathrm{midright}}(i-1), \vecd^{\mathrm{right}}(i-1), \vecy(i), \vecd^{\mathrm{left}}(i), \vecd^{\mathrm{midleft}}(i) \; .
\]
In doing this, we \emph{implicitly define} $\Ppadleft$ and $\Ppadright$ as the distributions of $\vecyleft(i)$ and $\vecyright(i)$. Before we proceed, we encourage the reader to validate the following points: $\vecyleft(i)$ and $\vecyright(i)$ are independent and their distributions
\begin{itemize}
    \item depend on the channel statistics of $W$;
    \item are not functions of $i$;
    \item are not functions of either $\vecx(i)$ nor $\vecy(i)$.
\end{itemize}

Consider an index $1 < i < \Phi$ (not the first nor last block). 
We now describe how $\vecyleft(i)$ depends on $\vecd^{\mathrm{midright}}(i-1)$ and $\vecd^{\mathrm{right}}(i-1)$. The description of how $\vecyright(i)$ depends on $\vecd^{\mathrm{left}}(i)$ and $\vecd^{\mathrm{midleft}}(i)$ is given by reflection symmetry. Before diving into the details, we emphasize that $\vecyleft(i)$ will consist of some suffix of $\vecd^{\mathrm{right}}(i-1)$. Since $\vecd^{\mathrm{right}}(i-1)$ is the result of sending a string of zeros, $\vecg^{\mathrm{right}}(i-1)$, we will indeed pad $\vecy(i)$ with a string of ``dirty zeros''. Here are the details.

\begin{enumerate}
    \item The genie considers the length of $\vecd^{\mathrm{midright}}(i-1)$.
        \begin{enumerate}
            \item \label{it:geniebad_shortones} If it is less than $h$, where
        \[
            h \triangleq \frac{\ell_{n_0} \cdot \beta}{4} \; ,
        \]
                the genie pads $\vecd^{\mathrm{midright}}(i-1)$ from the left. This is done by conceptually drawing a string from $p_{Y|X}(\cdot|1) $ and  prepending $\vecd^{\mathrm{midright}}(i-1)$ with the string. In practice, we use independent random variables to simulate $p_{Y|X}$. This is repeated until the length of $\vecd^{\mathrm{midright}}(i-1)$ is at least $h$.
        \end{enumerate}
    \item The genie considers the concatenated string
        \[
            \vecz = \vecd^{\mathrm{midright}}(i-1) \concat \vecd^{\mathrm{right}}(i-1) \; .
        \]
        It places a window of length $h$ at the right side of $\vecd^{\mathrm{midright}}(i-1)$. That is, the window starts at $z_s$ and ends at $z_e$, where $e = |\vecd^{\mathrm{midright}}(i-1)|$ and $s = e - h$.
    \item The genie draws a random integer $\rho = \rho^{\mathrm{left}}(i)$ uniformly from $\{1,2,\ldots,h\}$. We think of $\rho$ as a ``random dither''.
    \item The genie shifts the window by $\rho$ positions right. That is, $\rho$ is added to both $s$ and $e$.
        \begin{enumerate}
            \item \label{it:geniebad_shortFirstWindow} If the window falls off $\vecz$, that is, if $e > |\vecz|$, the genie chooses $\vecyleft = \epsilon$, the empty string. Otherwise, the genie continues to the next step.
        \end{enumerate}
    \item The genie counts the number of `$0$' symbols in the window (i.e., the cardinality of $\{s \leq i \leq e | z_i = 0\}$), 
    \item If the count is at least $h/2$, the genie sets $\vecyleft$ to the remainder of $\vecz$ after deleting $z_1$  to $z_e$ and then finishes by returning $\vecyleft$. 
    \item Otherwise, the genie shifts the window one frame right. That is, $h$ is added to both $s$ and $e$.
        \begin{enumerate}
            \item \label{it:geniebad_shortSecondWindow} If the window falls off $\vecz$, that is, if $e > |\vecz|$, the genie chooses $\vecyleft = \epsilon$, the empty string. Otherwise, the genie continues to the next step.
        \end{enumerate}
    \item We set $\vecyleft$ to the remainder of $\vecz$ after deleting $z_1$  to $z_e$ and then finish by returning $\vecyleft$.
\end{enumerate}

The rationale of above procedure will become clearer after we explain Aladdin's algorithm. For now, note that it is well defined and does indeed satisfy the requirements stated previously. The reader should also keep in mind that getting into a substep is `bad' with respect to Aladdin's ability to mimic the genie. That is, we would like the probability of entering substeps \ref{it:geniebad_shortones}, \ref{it:geniebad_shortFirstWindow}, or \ref{it:geniebad_shortSecondWindow} to be `small'.

We must address one last point: how the paddings for blocks $i=1$ and $i=\Phi$ are handled. The right padding for $i=1$ and the left padding for $i=\Phi$ are as above. The left padding for $i=1$ and the right padding for $i=\Phi$ (i.e., the edge padding) are given by random sampling from $\Ppadleft$ and $\Ppadright$. These choices are coupled so that the genie and Aladdin always choose the same realizations for these edge paddings.

\subsection{Aladdin parsing}
Aladdin receives the vector $\vecy$, and as a preliminary step adds the edge padding on the left and right (both of which are coupled to the genie's choices).
We denote the resulting vector $\vecy_{\textrm{pad}}$. 
Aladdin's parsing is given by $\mimicGenie(\vecy_{\textrm{pad}},n)$ where the recursive function $\mimicGenie(\vecz,m)$ is defined by:
\begin{itemize}
    \item If $m = n_0$, Aladdin returns $\vecz$. Otherwise,
    \item Aladdin builds $\zI$ and $\zII$ as follows, and then return $\mimicGenie(\zI,m-1)$ (which contains $2^{m-n_0-1}$ vectors), followed by $\mimicGenie(\zII,m-1)$ (which also contains $2^{m-n_0-1}$ vectors). Namely, Aladdin returns $2^{m-n_0}$ vectors.
    \item Let $\zI$ be the left half of $\vecz$ and $\zII$ be the right half of $\vecz$ (in case $|\vecz|$ is odd, $\zI$ is longer than $\zII$, by one bit). Then, Aladdin trims $\zI$ and $\zII$.
    \item Trimming $\zI$ is the ``mirror image'' of trimming $\zII$, which is done as follows:
        \begin{itemize}
    \item Aladdin places a window of length $h$ at the start of $\zII$. That is, the window starts at $s=0$ and ends at $e=h$. If in any stage of the algorithm the window ``falls off $\zII$'', meaning that $e > |\zII|$, Aladdin declares failure.
    \item Aladdin randomly and uniformly chooses a random dither $\rho'$ uniformly from $\{1,2,\ldots,h\}$.
    \item Aladdin shifts the window $\rho$ positions right by adding $\rho$ to both $s$ and $e$.
    \item Aladdin checks if the window contains at least $h/2$ `$0$' symbols. If it does, the process continues to the next step. If it does not, Aladdin moves the window one frame to the right by adding $h$ to both $s$ and $e$, and then repeats this bullet point.
    \item Aladdin trims $\zII$ by removing the first $e$ symbols.
\end{itemize}
\end{itemize}

\subsection{Connections between Aladdin and genie parsing}
Let us compare Aladdin's parsing to that of the genie. First of all, note that both decoders use $2\Phi-2$ random dithers during their runs (recall that we've denoted a random dither as $\rho$ for the genie and $\rho'$ for Aladdin). To help Aladdin match the genie, these dithers can be \emph{coupled}. That is, for each choice of random dithers the genie makes we couple a unique choice of random dithers that Aladdin makes.
We also couple the choices of  
The utility of this coupling is that, with high probability, both the genie and Aladdin return the same vector of DZP channel outputs $(\vecy^\star(i))_{i=1}^\Phi$.

Before describing the coupling, we note that, given the DZP parsing, the proof of Theorem~\ref{theo:main} follows essentially the same steps as the main result in~\cite{TPFV:20a}.
The steps will be detailed in a forthcoming longer version of this paper \cite{longer}.


For brevity, we explain the coupling in terms of just two dithers. Consider the $\rho$ that the genie chooses for padding $\vecx(i)$ from the left, for $i=\Phi/2+1$. We couple this $\rho$ with the $\rho'$ Aladdin chooses in the topmost part of the recursion, for producing $\zII$. Typically, the midpoint of $\vecz$ is in $\vecd^{\mathrm{midright}}(i-1)$ or $\vecd^{\mathrm{midleft}}(i-1)$. Aladdin adds the dither $\rho$ to the window, and then shift it one frame right, until the number of zeros is large enough. By Lemma~\ref{lemm:windowAdversary} we conclude that the number of zeros will typically not be large enough, until the window contains some part of $\vecd^{\mathrm{right}}(i-1)$. Consider the first time this happens, and set the Genie's $\rho$ to the number of symbols from $\vecd^{\mathrm{right}}(i-1)$. One can think of the genie as having a `shortcut' that avoids the previous steps Aladdin took. Both Aladdin and the Genie have the same window, at this point. If it contains enough zero symbols, they return the same padding. If it does not, they both shift it one frame right. At this stage, typically, the window will only contain symbols from $\vecd^{\mathrm{right}}(i-1)$, ``dirty zeros''. Thus, again by Lemma~\ref{lemm:windowAdversary}, Aladdin will typically stop at this stage, as the genie always does, and both will return the same left padding.

\appendix

\subsection{Proof of Lemma~\ref{lemm:windowAdversary}}

    Let $0 < \delta < 1$ be a constant, dependent on the channel, that we will fix later. Recall that the expected length of an output corresponding to a single input is $\beta$, and define
    \begin{equation}
        \label{eq:kHoeffding}
        k = \lceil h(1-\delta)/\beta \rceil \; .
    \end{equation}
    We can think of the output $\vecY$ as being manufactured as follows. We input the first bit ($x$) to the channel, then input $k$ more bits (all $x$), and if the output up to this point has length less than $h+1$, inputting however many bits (all $x$) are needed in order for the output length to be at least $h+1$. Then, we possibly remove the first bit of the output, and set $\vecY$ to the first $h$ bits. We will call the output corresponding to the $k$ input bits after the first input bit the \emph{essential output}. Our proof hinges on showing that the following two events occur with very high probability: 1) all of the essential output is contained in $\vecY$, and 2) the essential output has more than $h/2$ bits equal to $x$.

    Denote by $\vecZ = \vecZ_1 \concat \vecZ_2\concat \cdots \concat \vecZ_k$ the essential output, where $\vecZ_i$ is the output corresponding to input bit $i+1$. We find it easier to define bad events: event $A$ occurs if the length of $\vecZ$ is at least $h-1$; event $B$ occurs if $\vecZ$ contains at most $h/2$ bits equal to $x$. Clearly, if neither $A$ nor $B$ occur, the above good events occur\footnote{Recall that the output due to the first input bit has length at most $2$.}, and we correctly guess $x$.

    Now, let us choose
    \begin{equation}
        \label{eq:deltaHalfGamma}
        \delta = \min\left\{\frac{\gamma }{2 \cdot \alpha_{0|0}},\frac{\gamma }{2 \cdot \alpha_{1|1}},\frac{1}{2}\right\} \; ,
    \end{equation}
    where $\gamma$ is defined in (\ref{eq:gamma}). Note that, indeed, $0 < \delta < 1$. Let
    \begin{equation}
        \label{eq:hzeroprime}
        h_0' = \frac{2(\beta+1)}{\delta} - 1 > 1 \; ,
    \end{equation}
    where the inequality follows from (\ref{eq:deltaHalfGamma}). Assume that $h \geq h_0'$.
    
   By Hoeffding's bound \cite[Theorem 4.12]{MitzenmacherUpfal:17b},
    \begin{multline*}
        P(A) \leq 2 e^{\frac{-2 k \left( \frac{h-1}{k} - \beta \right)^2 }{4} } =  2 e^{\frac{-k \left( \frac{h-1}{k} - \beta \right)^2 }{2} } \\
             \eqann[\leq]{a}  2 e^{\frac{-\left(\frac{h(1-\delta)}{\beta} \right) \cdot  \left( \frac{h-1}{k} - \beta \right)^2 }{2} } \; ,
\end{multline*}
where \eqannref{a} follows from (\ref{eq:kHoeffding}). Noting the squared term on the RHS, we next show that
\[
    \frac{h-1}{k} >  \frac{h-1}{\frac{h(1-\delta)}{\beta} + 1} \geq \frac{\beta}{1-\delta/2} > \beta \; .
\]
Indeed, the first inequality follows from (\ref{eq:kHoeffding}), noting that $h-1$ is positive, since $h \geq h_0' > 1$; the second follows from (\ref{eq:hzeroprime}), recalling that $h \geq h_0'$; the third follows since $\delta > 0$. Thus, from the above two displayed equations we conclude that
\begin{equation}
    \label{eq:PAbound}
    P(A) < 2 e^{\frac{-h\left(\frac{(1-\delta)}{\beta} \right) \cdot  \left( \frac{\beta}{1-\delta/2} - \beta \right)^2 }{2} } \; .
\end{equation}

For event $B$, we use Hoeffding's inequality and (\ref{eq:kHoeffding}) to get
    \begin{multline*}
        P(B) \leq 2 e^{ \frac{-2 k \left( \alpha_{x|x} - \frac{h}{2k} \right)^2 }{4}  } = 2 e^{ \frac{-k \left( \alpha_{x|x} - \frac{h}{2k} \right)^2 }{2}  } \\
             \leq 2 e^{ \frac{-\left(\frac{h(1-\delta)}{\beta} \right) \cdot  \left( \alpha_{x|x} - \frac{h}{2k} \right)^2 }{2}  } \; .
    \end{multline*}
Focusing on the squared term on the RHS, we now prove that
\[
    \alpha_{x|x} \geq \frac{\beta/2}{1-\delta} \geq \frac{h}{2k} \; .
\]
The second inequality follows easily from (\ref{eq:kHoeffding}). For the first inequality, first recall that $\beta = \alpha_{x|x} + \alpha_{1-x|x}$, by (\ref{eq:beta}). Thus, it suffices to prove that
\(
    \frac{\alpha_{x|x} - \alpha_{1-x|x}}{2} \geq \delta \alpha_{x|x}  \; .
\)
By (\ref{eq:gamma}), this will follow if we prove that $\gamma \geq \delta \alpha_{x|x}$, which holds by (\ref{eq:deltaHalfGamma}). Thus, from the above two displayed equations we conclude that
\begin{IEEEeqnarray*}{rCl}
    P(B) &\leq& 2 e^{ \frac{-\left(\frac{h(1-\delta)}{\beta} \right) \cdot  \left( \alpha_{x|x} - \frac{\beta/2}{1-\delta} \right)^2 }{2}  } \\
         &= & 2 e^{ \frac{-\left(\frac{h(1-\delta)}{\beta} \right) \cdot  \left( \frac{ \frac{\alpha_{x|x} - \alpha_{1-x|x}}{2} - \delta \alpha_{x|x}}{1-\delta} \right)^2  }{2}} \; .
\end{IEEEeqnarray*}
Slightly refining the above arguments, we get from (\ref{eq:beta}), (\ref{eq:gamma}), and (\ref{eq:deltaHalfGamma}), that  
\[
    \frac{\alpha_{x|x} - \alpha_{1-x|x}}{2} \geq \gamma \geq \frac{\gamma}{2} \geq \delta \alpha_{x|x} \; .
\]
Thus, from the above two displayed equations we get that
\begin{equation}
    P(B) \leq 2 e^{ \frac{-\left(\frac{h(1-\delta)}{\beta} \right) \cdot  \left( \frac{ \gamma - \gamma/2 }{1-\delta} \right)^2  }{2}} 
    =  2 e^{ \frac{-h\left(\frac{1-\delta}{\beta} \right) \cdot  \left( \frac{ \gamma/2 }{1-\delta} \right)^2  }{2}} \; .  \label{eq:PBbound}
\end{equation}


In light of (\ref{eq:PAbound}) and (\ref{eq:PBbound}), let us define
\[
    c_0 \triangleq \frac{1}{2} \min\left\{\underbrace{\scriptstyle{\frac{\left(\frac{(1-\delta)}{\beta} \right) \cdot  \left( \frac{\beta}{1-\delta/2} - \beta \right)^2 }{2}}}_{c_0'},\underbrace{\scriptstyle{\frac{\left(\frac{1-\delta}{\beta} \right) \cdot  \left( \frac{ \gamma/2 }{1-\delta} \right)^2  }{2}}}_{c_0''} \right\}
\]
and take $h_0 \geq h_0'$ large enough such that for all $h \geq h_0$,
\[
    2 e^{-h \cdot  c_0'} + 2 e^{-h \cdot c_0''} \leq e^{-h \cdot c_0} \; .
\]
That is, $h_0 \triangleq \max\{ h_0', \ln(4)/c_0  \}$.

\IEEEtriggeratref{12}

\twobibs{
\bibliographystyle{IEEEtran} 
\bibliography{mybib.bib} 

\begin{thebibliography}{10}
\providecommand{\url}[1]{#1}
\csname url@samestyle\endcsname
\providecommand{\newblock}{\relax}
\providecommand{\bibinfo}[2]{#2}
\providecommand{\BIBentrySTDinterwordspacing}{\spaceskip=0pt\relax}
\providecommand{\BIBentryALTinterwordstretchfactor}{4}
\providecommand{\BIBentryALTinterwordspacing}{\spaceskip=\fontdimen2\font plus
\BIBentryALTinterwordstretchfactor\fontdimen3\font minus
  \fontdimen4\font\relax}
\providecommand{\BIBforeignlanguage}[2]{{%
\expandafter\ifx\csname l@#1\endcsname\relax
\typeout{** WARNING: IEEEtran.bst: No hyphenation pattern has been}%
\typeout{** loaded for the language `#1'. Using the pattern for}%
\typeout{** the default language instead.}%
\else
\language=\csname l@#1\endcsname
\fi
#2}}
\providecommand{\BIBdecl}{\relax}
\BIBdecl

\bibitem{Gallager_1961}
R.~Gallager, ``Sequential decoding for binary channels with noise and
  synchronization errors,'' 1961, {L}incoln Lab Group Report.

\bibitem{Dobrushin_1967}
R.~L. Dobrushin, ``Shannon's theorems for channels with synchronization
  errors,'' \emph{Problemy Peredachi Informatsii}, vol.~3, no.~4, pp. 18--36,
  1967.

\bibitem{Davey_2001}
M.~C. Davey and D.~J. MacKay, ``Reliable communication over channels with
  insertions, deletions, and substitutions,'' \emph{IEEE Trans.\ Inform.\
  Theory}, vol.~47, no.~2, pp. 687--698, 2001.

\bibitem{Mitzenmacher_2009}
M.~Mitzenmacher, ``A survey of results for deletion channels and related
  synchronization channels,'' \emph{Probability Surveys}, vol.~6, pp. 1--33,
  2009.

\bibitem{Kanoria_2013}
Y.~Kanoria and A.~Montanari, ``Optimal coding for the binary deletion channel
  with small deletion probability,'' \emph{IEEE Trans.\ Inform.\ Theory},
  vol.~59, no.~10, pp. 6192--6219, 2013.

\bibitem{Rahmati_2015}
M.~Rahmati and T.~M. Duman, ``Upper bounds on the capacity of deletion channels
  using channel fragmentation,'' \emph{IEEE Trans.\ Inform.\ Theory}, vol.~61,
  no.~1, pp. 146--156, 2015.

\bibitem{Castiglione_2015}
J.~Castiglione and A.~Kavcic, ``Trellis based lower bounds on capacities of
  channels with synchronization errors,'' in \emph{Information Theory
  Workshop}.\hskip 1em plus 0.5em minus 0.4em\relax Jeju, South Korea: IEEE,
  2015, pp. 24--28.

\bibitem{Cheraghchi_2019}
M.~Cheraghchi, ``Capacity upper bounds for deletion-type channels,''
  \emph{Journal of the ACM (JACM)}, vol.~66, no.~2, p.~9, 2019.

\bibitem{Tal-isit19}
I.~Tal, H.~D. Pfister, A.~Fazeli, and A.~Vardy, ``Polar codes for the deletion
  channel: Weak and strong polarization,'' in \emph{Proc.\ IEEE Int.\ Symp.\
  Inform.\ Theory}, 2019, pp. 1362--1366.

\bibitem{TPFV:20a}
------, ``Polar codes for the deletion channel: weak and strong polarization,''
  2020, preprint arXiv:1904.13385v2.

\bibitem{Wang_2014}
R.~Wang, R.~Liu, and Y.~Hou, ``Joint successive cancellation decoding of polar
  codes over intersymbol interference channels,'' 2014, preprint
  arXiv:1404.3001.

\bibitem{Wang_2015}
R.~Wang, J.~Honda, H.~Yamamoto, R.~Liu, and Y.~Hou, ``Construction of polar
  codes for channels with memory,'' in \emph{2015 IEEE Information Theory
  Workshop}, October 2015, pp. 187--191.

\bibitem{Thomas_2017}
E.~K. Thomas, V.~Y.~F. Tan, A.~Vardy, and M.~Motani, ``Polar coding for the
  binary erasure channel with deletions,'' \emph{IEEE Communications Letters},
  vol.~21, no.~4, pp. 710--713, April 2017.

\bibitem{Tian_2017}
K.~Tian, A.~Fazeli, A.~Vardy, and R.~Liu, ``Polar codes for channels with
  deletions,'' in \emph{55th Annual Allerton Conference on Communication,
  Control, and Computing}, 2017, pp. 572--579.

\bibitem{Tian_2018}
K.~Tian, A.~Fazeli, and A.~Vardy, ``Polar coding for deletion channels: Theory
  and implementation,'' in \emph{IEEE International Symposium on Information
  Theory}, 2018, pp. 1869--1873.

\bibitem{Tian_it2018}
------, ``Polar coding for deletion channels,'' 2018, submitted to IEEE Trans.\
  Inform.\ Theory.

\bibitem{Li_2019}
Y.~Li and V.~Y.~F. Tan, ``On the capacity of channels with deletions and
  states,'' 2019, preprint arXiv:1911.04473.

\bibitem{SasogluTal:18a}
E.~\c{S}a\c{s}o\u{g}lu and I.~Tal, ``Polar coding for processes with memory,''
  \emph{IEEE Trans.\ Inform.\ Theory}, vol.~65, no.~4, pp. 1994--2003, April
  2019.

\bibitem{ShuvalTal:18a}
B.~Shuval and I.~Tal, ``Universal polarization for processes with memory,''
  2018, preprint arXiv:1811.05727v1.

\bibitem{Shuval_Tal_Memory_2017}
------, ``Fast polarization for processes with memory,'' \emph{IEEE Trans.\
  Inform.\ Theory}, vol.~65, no.~4, pp. 2004--2020, April 2019.

\bibitem{longer}
H.~D. Pfister and I.~Tal, ``Polar codes for channels with insertions,
  deletions, and substitutions,'' arXiv preprint in preparation.

\bibitem{MitzenmacherUpfal:17b}
M.~Mitzenmacher and E.~Upfal, \emph{Probability and Computing: Randomizition
  and Probabilistic Techniques in Algorithms and Data Analysis}, 2nd~ed.\hskip
  1em plus 0.5em minus 0.4em\relax Cambridge, UK: Cambridge University Press,
  2005.

\end{thebibliography}
}
{

}

\end{document}